\documentclass[a4paper,UKenglish,cleveref,  thm-restate]{lipics-v2021}
\pdfoutput=1 
\hideLIPIcs
\nolinenumbers

\usepackage{booktabs} 
\usepackage[ruled]{algorithm2e} 

\SetAlFnt{\small}
\SetAlCapFnt{\small}
\SetAlCapNameFnt{\small}
\SetAlCapHSkip{0pt}
\IncMargin{-\parindent}

\usepackage{cleveref}

\crefname{lemma}{Lemma}{Lemmas}
\crefname{theorem}{Theorem}{Theorems}
\usepackage{graphicx} 

\usepackage{amsthm}
\usepackage{amsmath}
\usepackage{amsfonts}
\usepackage{mdframed}
\usepackage[colorinlistoftodos]{todonotes}
\usepackage{verbatim} 
\usepackage{fetamont}
\usetikzlibrary{patterns}
\usepackage{caption, subcaption}
\usepackage{xspace}
\usepackage[nolist]{acronym}
\usepackage{ifthen}
\usepackage{braket}
\usepackage[nolist]{acronym}
\usepackage{tikz}
\usepackage{pgfplots}\pgfplotsset{compat=1.17}
\usepackage{nicefrac}

\usepackage{etoolbox}
\AtBeginEnvironment{align}{\setcounter{equation}{0}}

\newtheorem{metatheorem}{Meta-Theorem}

\newcommand{\probName}{MMS scheduling with deadlines problem}
\newcommand{\jobs}{\mathcal{J}}
\newcommand{\machines}{\mathcal{M}}
\newcommand{\fairValue}{\varphi}
\newcommand{\partition}{\textsc{Partition}}
\newcommand{\ecpartition}{\textsc{Equal-Cardinality Partition}}

\newcommand{\rtNfold}{T(\text{$n$-fold})}
\newcommand{\MMSi}{\text{MMS}_i}
\newcommand{\multmms}{\text{$\times$-}}
\newcommand{\addmms}{\text{+-}}
\newcommand{\wfmms}{\text{wf-}}

\authorrunning{K. Jansen, A. Lassota, M. Tutas and A. Vetta} 

\Copyright{K. Jansen, A. Lassota, M. Tutas and A. Vetta} 


\author{Klaus Jansen}{Kiel University}{kj@informatik.uni-kiel.de}{}{}
\author{Alexandra Lassota}{Technical University Eindhoven}{a.a.lassota@tue.nl}{}{}
\author{Malte Tutas}{Kiel University}{mtu@informatik.uni-kiel.de}{}{}
\author{Adrian Vetta}{McGill University}{adrian.vetta@mcgill.ca}{}{}
\title{FPT Algorithms using Minimal Parameters for a Generalized
  Version of Maximin Shares}

\begin{document}

\maketitle
 \begin{abstract}
We study the computational complexity of fairly allocating indivisible, mixed-manna items. For basic measures of fairness, this problem is hard in general. 
Thus, research has flourished concerning input classes 
where efficient algorithms exist, both for the purpose of
establishing theoretical boundaries and for the purpose
of designing practical algorithms for real-world instances. 
Notably, the paradigm of fixed-parameter tractability~(FPT) has lead to new insights and improved algorithms for a variety of fair allocation problems; see, for example, Bleim et al. (IJCAI’16), Aziz et al. (AAAI’17), Bredereck et al. (EC'19) and Kulkarni et al. (EC’21). 

Our focus is the fairness measure {\em maximin shares} (MMS). 
Motivated by the general non-existence of MMS allocations, Aziz et al. (AAAI’17) studied {\em optimal MMS allocations}, namely
solutions that achieve the best $\alpha$-approximation for the maximin share value of every agent. These allocations are guaranteed to exist, prompting the important open question of whether optimal MMS allocations can be computed efficiently.  
We answer this question affirmatively by providing FPT algorithms to output optimal MMS allocations. Furthermore, 
our techniques extend to find allocations that optimize alternative objectives, such as minimizing the additive approximation, and maximizing some variants of global welfare.

In fact, all our algorithms are designed for a more general MMS problem in machine scheduling. 
Here, each mixed-manna item (job) must be assigned to an
agent (machine) and has a processing time and a deadline.
This generalization is of independent interest as it induces
models for important applications such as ride-hailing services and shift scheduling.
We develop efficient algorithms running in FPT time. Formally, we present polynomial time algorithms w.r.t.\ the input size times some function dependent on the parameters that yield optimal maximin-value approximations (among others) when parameterized by \\
(i) the number of items, \\
(ii) the maximum number of items in any group (each group is independently scheduled/assigned on the machines/agents and only the overall valuation ties them together), the number of agents, and the maximum absolute valuation, \\
(iii) the number of agents, the number of different deadlines, the maximum processing time, and the maximum absolute valuation.

To obtain these results, we utilize machinery from graph theory, dynamic programming, and integer programming. In particular, for the latter, we employ the framework of $n$-fold integer programs -- these are integer programs with a specific block structure that allows them to be solved in FPT time. 
Our results are tight.
We prove that \emph{any} stronger parameterization is NP-hard even if all parameters only have constant values. Thus, assuming P $\neq$ NP, there exists no FPT time algorithm for the respective stronger parameterization. Further, the running time in (i) is tight assuming the exponential time hypothesis. 
Consequently, this work gives a complete complexity theoretical picture and efficient algorithms for fairly allocating indivisible, mixed-manna items under the notion of maximin share fairness.
\end{abstract}

\section{Introduction}
Take the classical {\em job scheduling} problem, where we are given
a collection of $n$ jobs with individual processing times and deadlines.
What is the best way to assign the jobs to a set of machines?  
This question has been intensively studied, but we want to go even further: can we find an assignment of the jobs in a manner that is {\em fair} to the machines? 
Here the reader may protest: {\em android rights} might be an important ontological conundrum but, despite the recent spectacular advances in artificial intelligence, a dilemma that perhaps one can defer to the future? 
However, exchange machines for people and suddenly the job scheduling problem with deadlines proffers a simple model for huge chunks of the gig economy today where
fairness does matter. Consider for example {\em ride-hailing systems}
which generate tens of billions of rides annually worldwide.
To date, the primary objectives in the allocation of rides to drivers have been
consumer satisfaction and maximizing sales volume and platform 
profitability. Fairness to drivers has been, at best, a secondary consideration,
leading to large disparities in driver satisfaction and 
surveys suggesting low wages; for example, Zoepf et. al~\cite{ZCA18} 
declare up to three-quarters of Uber drivers earn below the minimum wage and one-third obtain incomes insufficient to cover costs and vehicle depreciation. Similar fairness issues arise for {\em food delivery platforms}~\cite{SDC23}. Consequently, the design of fairer job 
allocation mechanisms for applications in the gig economy is now receiving
considerable attention~\cite{SBZ19, LZB19, BH20, SJY22, GNC22}.

This motivates our study of fair job scheduling with deadlines.
Of course, this begets the question of what exactly is ``fair''.
We take a theoretical approach and use the fairness measure of {\em maximin shares} due to Budish~\cite{Bud11} and, more specifically, {\em optimal maximin shares} due to Aziz et al.~\cite{ARS17}.
Our objective is to provide a parameterized characterization 
of exactly when polynomial times algorithms exist to find fair
schedules. Let's begin by formally describing the problem.

\subsection{The Model}
In the job scheduling with deadlines problem, there is a set $\machines$ of $m$ machines and set $\jobs$ of $n$ jobs.
Each job $j_t$ has a processing time $p_t\in \mathbb{Z}^+$ and a deadline
$d_t\in \mathbb{Z}^+$. 
An allocation $(\jobs_1,\jobs_2,\dots, \jobs_m)$ 
is a partition of the jobs over the machines. That is, a {\em bundle} of $\jobs_i$ of jobs is allocated to machine $m_i$, where $\jobs_{i}\cap\jobs_{j}=\emptyset$ for all $i\neq j \in [m]$ and $\bigcup_{i\in [m]}\jobs_{i}=\jobs$. The allocation is {\em feasible} if every job is completed before its deadline. Bundles can easily be checked for feasibility using the {\tt earliest deadline first} (EDF) rule, which sorts the jobs in each bundle according to their deadlines and places them in that order on the machine. Specifically, there exists a sequential ordering the jobs on each machine $m_i$ such that, for every job $j_t\in \jobs_i$, the sum of the processing times of the jobs before and including $j_t$ is at most the deadline $d_t$. We denote by $\mathcal{F}$
the set of feasible allocations.

To evaluate the fairness of an allocation, let machine $m_i$ have a valuation $v_i(j_t)\in \mathbb{Z}$ for job $j_t$.
We allow for {\em mixed-manna} valuations: the job may either be
a {\em good} where $v_i(j_t)\ge 0$ or a {\em chore} ({\em bad}) where $v_i(j_t)<0$.
Note that a job may simultaneously be a good for one machine but a chore for
another machine. We assume the agents have additive valuations functions;
that is, $v_i(\jobs_i) = \sum_{j\in \jobs_i} v_i(j)$.

We remark that the job scheduling with deadlines problem generalizes the 
classical fair division problem where $n$ {\em items} must be allocated among $m$ {\em agents}. To see this, simply associate each agent with a machine and each item with a job.\footnote{Typically in the fair division problem there are $m$ items and $n$ agents. We switch this notation to be consistent with the scheduling literature where there are $n$ jobs and $m$ machines.}
Now set $p_t=d_t=0$, for every item/job $j_t$. Then any allocation is feasible and we may assign the items to agents in any way we desire.

As stated, the fairness measure we use is that of maximin shares~\cite{Bud11}.
Assume machine (agent) $m_i$ partitions the jobs (items) into 
$m$~bundles that form a feasible allocation. The other machines choose one bundle each, in turn, and then machine $m_i$ receives the last remaining bundle. Intuitively, a risk averse machine (agent) desires a feasible allocation that maximizes the value of its least desired bundle in the allocation. The minimum value of a bundle in the optimal partition is called the {\em maximin share value} for the machine. Formally
$$\text{MMS}_i=\max_{(\jobs_1,\ldots,\jobs_m)\in \mathcal{F}} \, \min_{k\in [m]} \, v_i(\jobs_k).$$
As the machines have different valuation functions, the feasible allocations they each propose would differ. Consequently their maximin share values differ in general.

We say that $(\jobs_1,\dots, \jobs_m)\in \mathcal{F}$
is an $(\alpha_1,\dots, \alpha_m)$-{\em (multiplicative) maximin shares allocation} if, for every machine $m_i$,
$$
v_i(\jobs_i) \ge
\begin{cases}
\alpha_i\cdot MMS_i &\quad \text{if } MMS_i\ge 0,\\
\frac{1}{\alpha_i}\cdot MMS_i &\quad \text{if } MMS_i< 0.
\end{cases}
$$
For the special case of $(\alpha_1,\dots, \alpha_m)=(1,\dots, 1)$ we say that $(\jobs_1,\dots, \jobs_m)\in \mathcal{F}$ is a {\em maximin shares allocation}.
The {\em maximin shares problem}, denoted MMS, asks whether or not a maximin shares allocation exists.
In fact, a maximim shares allocation need not exist even in the basic setting without scheduling constraints~\cite{KPW18}.
Hence, in this paper, the objective is to find a feasible allocation that
optimizes the approximation with respect to three natural objective functions. The first is
\begin{enumerate}
    \item[(1)] the {\em optimal (multiplicative) maximin shares problem}, denoted as \multmms MMS, where the aim is to find a feasible allocation that
maximizes the minimum of the $\alpha_i$:
$$\max_{(\jobs_1,\ldots,\jobs_m)\in \mathcal{F}} \, \min_{i\in [m]} \alpha_i$$
\end{enumerate}
This objective was proposed by Aziz et al.~\cite{ARS17} and can be interpreted as {\em egalitarian (Rawlsian) welfare}, assuming that 
the utilities are normalized by their maximin share values.
The remaining two objectives correspond to {\em social welfare}. But the correspondence becomes cleaner when we switch to additive rather than 
multiplicative bounds. 

Accordingly,
we say $(\jobs_1,\dots, \jobs_m)\in \mathcal{F}$
is a $(\delta_1,\dots, \delta_m)$-{\em (additive) maximin shares allocation} if, for every machine $m_i$,
we have $\delta_i\ge 0$ and $v_i(\jobs_i) \ge MMS_i - \delta_i$.
This induces the following two optimization problems:
\begin{enumerate}
    \item[(2)] the \emph{optimal (additive) maximin shares problem}, denoted as \addmms MMS, where the aim is to find a feasible $(\delta_1,\dots, \delta_m)$-(additive) maximin shares allocation that minimizes the maximum of the $\delta_i=\max\{(MMS_i-v_i(J_i)) ,  0\}$:
    $$\min_{(\jobs_1,\ldots,\jobs_m)\in \mathcal{F}} \, \max_{i\in [m]} \delta_i$$
\item[(3)]  the {\em optimal (welfare) maximin shares problem}, denoted as \wfmms MMS, where the 
aim is to find a feasible $(\delta_1,\dots, \delta_m)$-(additive) maximin shares allocation that minimizes the sum of the $\delta_i=\max\{(MMS_i-v_i(J_i)) ,  0\}$:
$$\min_{(\jobs_1,\ldots,\jobs_m)\in \mathcal{F}} \, \sum_{i\in [m]}  \delta_i$$
\end{enumerate}
Given our scheduling framework, we colloquially refer to this set of problems as MMS, \multmms MMS, \addmms MMS, and \wfmms MMS scheduling with deadlines. We note a connection with the Santa-Claus problem, which corresponds to a scheduling problem with target $T$ on unrelated machines (machine-dependent processing times), where valuations equal processing times, and which asks for an assignment of tasks such that every machine exceeds load $T$. However, our problems differ as here processing times and valuations are {\em independent}.

We remark that our mechanism also allows for the weighting of agents. So let us briefly comment on the distinction between multiplicative and additive guarantees and on weighting agents. In approximation theory, additive guarantees are typically harder to achieve than multiplicative guarantees and are, thus, considered preferable. However, an allocation mechanism based upon additive approximation is incompatible with scale-freeness, opening up the possibility that an agent may manipulate it by scaling its valuation function. In practice, though, most allocation mechanisms are insulated against such manipulations. The most natural way this is accomplished is where the mechanism itself ``scales'' the valuation functions, e.g, valuation functions may be prescribed to be constant-sum. Notable illustrations of this are sports drafts, such as bidding budgets in the IPL auction in cricket or total salary caps in US sports. Generally, bidding and/or priority allocation mechanisms do scale/weight agents e.g. in assigning students to courses, job shifts to employees, or landing slots to airlines. Priority-scaling is also applicable in e.g. estate division, and budget-dependent-scaling in e.g. room allocation in shared housing. 

However, these allocations problems are all hard. Indeed it is hard to just compute the maximin share value $\MMSi$ for a single machine $m_i$~\cite{Woe97, BL16, ARS17}. The purpose of this paper is to circumvent these discouraging hardness results using a parameterized complexity approach. Indeed, not only do we provide efficient algorithms to compute the maximum share values but also 
efficient algorithms to compute allocations that solve the 
optimal maximin shares and maximum welfare maximin shares problems.

We consider six natural parameters associated with the 
job scheduling with deadlines problem.
The two most basic are the number of jobs $n$ and
the number of machines $m$.  
Three more are $p_{\max}$, the maximum processing time of any job, $d_{\max}$, the maximum deadline of any job, and $v_{\max}$,
the maximum absolute value any machine has for a job.
In fact, $p_{\max}\le d_{\max}$. Therefore, rather than $d_{\max}$, we instead consider
the parameter $\#d$, the number of distinct deadlines; this is a stronger parameterization because
$\#d \le d_{\max}$.
The final parameter is $g_{\max}$, the maximum number of jobs in any group of jobs, motivated by applications where jobs inherently belong to distinct groups that can be scheduled independently. For example,
suppose that jobs must be scheduled within distinct time-shifts, such as days. Then jobs within the same day form a group (implicitly, here the jobs have an arrival time as well as a deadline that fall within the same day).
We remark, that whilst groupings simplify the ``daily'' scheduling problem, the fairness criterion remains complex as it must apply as a sum over all days (groups).

With the parameterization framework complete, we may now describe our results.

\subsection{Our Results}

An algorithm is {\em fixed parameter tractable} (FPT) if it runs in time polynomial in the input size $|I|$ multiplied by a function of the
specified parameters. In this paper, we classify exactly when FPT algorithms exist for the MMS scheduling with deadlines problem, that is, efficiently compute allocations that optimize $(\alpha_1, \dots, \alpha_m)$ 
or $(\delta_1, \dots, \delta_m)$ for our three optimal maximin shares problems, \multmms MMS, \addmms MMS, and \wfmms MMS. Recall our six parameters are
$\{n, m, \#d, p_{\max}, v_{\max}, g_{\max}\}$.
That is, in order, the number of jobs, 
the number of machines, the number of distinct deadlines, the maximum processing time of any job, the maximum absolute value, and the maximum number of jobs in a group.

We may state our main results concisely using two "meta-theorems".
The first gives three parameterization settings where FPT algorithms exist.
\begin{metatheorem}\label{meta:1}
The MMS, \multmms MMS, \addmms MMS, and \wfmms MMS job scheduling with deadlines problems are fixed-parameter tractable when parameterized by (i) $n$, (ii) $\{m, v_{\max},g_{\max}\}$, or (iii) $\{m, \#d, p_{\max}, v_{\max}\}$.
\end{metatheorem}
As key tools in the design of algorithms for these problems we use maximal matchings, dynamic programming and the theory of $n$-fold integer programming.
We further outline adaptations to our frameworks to admit the rejection of jobs, which incurs some job dependent costs.

The second meta-theorem states that these three are the only minimal parameterizations. {\em Any} stronger parameterization, that is, parameterizing upon a strict subset of the parameters, is NP-hard even if the parameters have constant value. 
\begin{metatheorem}\label{meta:2}
Assuming $P\neq NP$, no fixed-parameter tractable algorithms exist for 
any stronger parameterization than (i) $n$, (ii) $\{m, v_{\max},g_{\max}\}$, or (iii) $\{m, \#d, p_{\max}, v_{\max}\}$ for the MMS, \multmms MMS, \addmms MMS, and \wfmms MMS job scheduling with deadlines problems.
\end{metatheorem}

Thus, assuming P $\neq$ NP, there exists no FPT time algorithm for the respective smaller set of parameters.
We show this by giving reductions from NP-hard problems to our respective problem such that the values of the parameters are constants. This proves that our FPT time algorithms are designed to use the \emph{minimal} number of parameters necessary to ensure fixed-parameter tractability. We enrich this line of results by proving that the running time for parameter $n$ is tight for each of the problems assuming the exponential time hypothesis (see Theorem~\ref{thm:ethLB}).

Observe that our polynomial time algorithms do allow 
either $n$ or $m$ to be large (but not both). This condition holds for most of the fair division applications described above. For example, in flight scheduling the number of landing slots $m$ is large but the number of airlines $n$ is small.

Overall, this work offers an understanding of the complexity theoretical landscape and efficient algorithms for fairly allocating indivisible, mixed-manna items under the notion of maximin share fairness enriched by scheduling constraints.

\subsection{Related Work}
The problem of fair division was formally introduced by Steinhaus~\cite{Ste48}. Classical works focused on the case of divisible items and {\em cake-cutting}~\cite{BK96,RW98}.
More recently, for practical reasons, there has been a plenitude of works 
on the practical case of indivisible items. For surveys, we refer the reader to~\cite{Aziz20, ALM22, AAB23, AT20}.

Among the fairness measures in the literature of particular relevance here is {\em proportionality}~\cite{Ste48}. An allocation of the items to the agents is {\em proportional} if, for every agent, the value that the agent has for the {\em grand bundle} (all of the items) is at most a factor $n$ times greater than the value it has for the bundle it receives. 
It is, though, easy to show that proportional allocations need not 
exist. Maximin shares fairness~\cite{Bud11} relaxes the concept of proportionality and originates with the work of Hill~\cite{Hill87}.
In fact, the idea behind maximin shares dates back millennia, at least as
far as the division of the land of Canaan. Indeed,   
maximin shares is a multiagent generalization of the two agent ``{\em I divide, you choose}" protocol implemented by Abraham and Lot~\cite{bible}.

Maximin share allocations, where every agent receives a bundle
whose value is at least its maximin share, are guaranteed to 
exist in the special case of two agents with additive valuations\cite{BL16}.
They also exist given assorted strict restrictions on the valuation functions~\cite{BL16, AMN17, EPS22}.
But, for general multiagent settings, maximin share allocations need not exist, even for additive valuation functions~\cite{KPW18, FST21}.
Given this negative result, Procaccia and Wang~\cite{PW14}
proposed studying {\em approximate} maximin share allocations. 
A line of research searched for the best possible approximation guarantees~\cite{KPW18, AMN17, BK20, GT21} with 
the best factor currently due to Akrami and Garg~\cite{AG24}.

Rather than approximations to maximin share values, Aziz et al.~\cite{ARS17} propose study of the {\em optimal maximin share value}.
Specifically, find the maximum $\alpha$, and an associated allocation, such that every agent receives a bundle worth at least an $\alpha$ factor of its
maximin share value. This {\em optimal maximin share allocation} is, by definition, guaranteed to exist. We remark that~\cite{ARS17} study the case of chores and note a link between maximin shares and parallel machine scheduling, when processing times equal values.

However, in general, it is hard to find an optimal maximin share allocation. 
Indeed, it is NP-complete to even compute
the maximin share value for a single agent~\cite{BL16, ARS17},
but a PTAS exists due to Woeginger \cite{Woe97}.
When the number of agents is {\em fixed}, a polynomial time approximation scheme (PTAS) for finding an
optimal maximin share allocation for chores can be obtained~\cite{ARS17}.
The paper \cite{HNN18} extends this to other fairness measures. 
The problem becomes even harder under additional constraints, such as
the maximim shares problem on a graph, where the items are vertices and feasible bundles must be connected components, ~see \cite{LT20, GS20}.

Kulkarani et al.~\cite{KMT21}, like us, consider the general case of mixed-manna. They study the optimal maximin share allocation problem with the
addition requirement of Pareto optimality. For restricted instances, they provide polynomial time approximation schemes. 

What about polynomial time {\em exact} algorithms?
Bleim et al.~\cite{BBN16} design FPT algorithms to compute Pareto efficient and envy-free allocations. But, in general, our knowledge on  
the topic of efficient algorithms is sparse.
To remedy this, our work is inspired by recent theoretical studies on $n$-fold integer programming~\cite{AschenbrennerH07,CslovjecsekEHRW21,CslovjecsekEPVW21,DeLoera2008,EisenbrandHK18,HemmeckeOR13,JansenLR20,KouteckyLO18}.
These $n$-fold integer programs have found wide application; see~\cite{ChenMYZ17,DeLoera2008,GavenciakKK22,JansenKMR22,KnopK18,KnopKLMO21,KnopKM20,KnopKM20b} for illustrative examples.
Of particular relevance here, Bredereck et al.~\cite{BKK19} propose the use of $n$-fold integer programs in computing fair allocations, including an FPT algorithm for the optimal maximin share allocation parameterized by $\{m, v_{\max}\}$. We continue the success story of using $n$-fold IPs for fair allocation problems, and greatly extend the results of~\cite{BKK19}.
We emphasize that this needs careful crafting of new integer programs, as the $n$-fold IPs of Bredereck et al. cannot be extended to solve our problem. In fact, we prove the parameterization provided in~\cite{BKK19} is not strong enough to make our more general problem tractable. We overcome several limitations of their model. In particular, (i)~\cite{BKK19} studies the relaxed case where not every item/job needs to be allocated, making the problem much easier, (ii)~they cannot handle mixed-manna valuations, and (iii)~they only show how to solve the case for MMS shares of value $1$, but cannot approximate the value, neither multiplicatively nor additively.

\section{Overview}
In this section, we give an overview of the paper and the techniques within.
The proof of Meta-Theorem~\ref{meta:1} consists of the design of three
FPT algorithms for the parameterizations $\{n\}$, $\{m, v_{\max},g_{\max}\}$ and $\{m, \#d, p_{\max}, v_{\max}\}$. These algorithms are given
in Sections~\ref{sec:n}, \ref{sec:mvnmax} and \ref{sec:mdvp}, respectively. The hardness results required to prove Meta-Theorem~\ref{meta:2}
are presented in Section~\ref{sec:hardness}. 
Then in Section~\ref{sec:extensions} we explain how our results extend to even more general job scheduling problems. These include applications where
it is not necessary that every job is scheduled, but where each job incurs penalty cost if it is discarded. 
We conclude in Section~\ref{sec:EF} with explaining why other measures such as envy-freeness and its relaxations are not conducive to FPT analyses for the job scheduling problem.

It is beneficial to briefly outline our approaches for proving
Meta-Theorems~\ref{meta:1} and~\ref{meta:2}. 
For all algorithms, we first assume that we are given a value for each machine that we need to either satisfy with the schedule, or show that it is not possible to achieve for all machines simultaneously. This is later used to compute the MMS values for each machine, and to solve the optimal maximin shares (\multmms MMS), optimal additive maximin shares (\addmms MMS), and optimal welfare maximin shares problems (\wfmms MMS).

Let us begin with the FPT algorithm with parameter $n$, the number of items. For the \multmms{} and the \addmms \probName, we present a dynamic progam that iterated through the number $k$, $k \in [m]$ of machines (in any order) and every possible subset $S \subset \jobs$ of jobs and checks, whether $S$ can be scheduled on the first $k$ machines achieving the desired MMS values. 

As for the \wfmms \probName{}, we apply elementary counting arguments, as having $n$ as a parameter, we can exhaustively test each {\em partition} of the jobs into \emph{bundles} $\jobs_1,\ldots,\jobs_m$.
We first test if the partition is feasible, that is, whether each bundle induces a schedule that respects the job deadlines.
If the partition is indeed feasible, it remains to assign the bundles to machines. That is, to turn the partition into a schedule.
To do this, we build a bipartite graph $G = (\{A, B\}, E)$
where the vertices in $A$ are the machines and the vertices in $B$ are the bundles.
The edge $(a,b), a \in A, b \in B$ has weight $w_{a,b} = \sum_{j_t \in b} v_a(j_t)$,  the valuation that machine $a$ has for bundle $b$. If the weight is below the desired valuation value for the machine
then we omit the corresponding edge from the graph. If the graph contains a perfect matching, each machine obtains the
desired target value. Otherwise, if this test fails for all partitions, then no 
suitable allocation exists for the desired MMS values.

Our FPT algorithms for the parameterizations $\{m, v_{\max},n_{max}\}$ and $\{m, \#d, p_{\max}, v_{\max}\}$ are more complex and utilize $n$-fold integer programming. Such $n$-fold IPs are a special subclass of integer programs whose feasibility constraints are of the form
\begin{align}
    \sum_{i=1}^{n} C_i y_i &\leq a \label{global} \\ 
    D_i y_i &\leq b_i \quad\forall i=1, \dots, N \label{local}
\end{align}
where $C_i \in \mathbb{Z}^{r \times t}$ and $D_i \in \mathbb{Z}^{s \times t}$
are matrices\footnote{If the matrices have fewer columns or rows then we can pad them with zeros to obtain this property.}, and $y = (y_1, \dots, y_N)$ is the sought solution.
In other words, if we delete the first $r$ constraints (\ref{global}) $\{\sum_{i=1}^{n} C_i y_i\le a\}$, called the \emph{global} constraints, then the problem decomposes into $N$ independent programs (\ref{local}) $\{D_i y_i \leq b_i, \forall i\}$, given by the \emph{local} constraints. Note that we use $N$ for the number of blocks in an IP, and $n$ for the number of jobs in the instance; however, we construct the $n$-fold IP such that $N=n$.

Now denote by $\Delta$ the largest absolute value in the $C_i$ and $D_i$ matrices. Cslovjecsek et al.~\cite{CslovjecsekEHRW21} proved $n$-fold IPs can be solved in time 
$(Nt)^{(1+o(1))} \cdot 2^{O(rs^2)}\cdot(rs\Delta)^{O(r^2s+s^2)}$.
Observe that this is an FPT algorithm when parameterized by the dimensions of the matrices $C_i$, $D_i$, and their largest entry $\Delta$. 
This means that the goal is to formulate our problems as $n$-fold integer programs where the dimensions of the $C_i$ and $D_i$ matrices, as well as the largest absolute value $\Delta$ are only dependent on our specified parameters. 

Consider now the parameterization $\{m, v_{\max},g_{\max}\}$ for an instance with groups, which we call \emph{layers}. For motivation here, let each layer consist of jobs that have to be scheduled on the same day. Thus each day (and thereby its jobs) are independent from each other except for calculating the overall sum of values to each machine. Hence, we handle the scheduling constraints, such as each job being feasibly assigned to exactly one machine for each group, using the local constraints~(\ref{local}). As the maximum number of jobs in a group $g_{\max}$ and the number of machines $m$ are parameters, the $D_i$ remain small enough in dimension. Given our target values for each machine, 
the global constraints~(\ref{global}) only encode the requirement that 
each machine achieves its target value over all days. This gives us $m$ global constraints defined by the $C_i$, in which the largest absolute valuation appears as a factor. Solving this $n$-fold IP gives the desired schedule for our problem.

Next, take the parameteriztion on $\{m, \#d, p_{\max}, v_{\max}\}$,  
the number of machines, the number of distinct deadlines, the maximum processing time, and the maximum absolute valuation. 
Here, the necessary constraint that each job needs to be assigned to a machine can be handled independently of the other jobs using the local constraints (\ref{local}). The few global constraints check for every machine whether or not the jobs allocated to them can be feasibly scheduled (introducing the maximal processing time and the number of different deadlines as factors) and whether or not the assigned jobs achieve the desired target value (introducing the largest absolute valuation as a factor).

One core strength of the algorithmic results given above is their adaptability. By embedding the algorithms into different frameworks, we can achieve several different objectives, all in FPT time. On a basic level, we can maximize the lowest value attained by a machine. This directly allows us to calculate the maximin-share of every machine. With these values in hand, we can decide whether it is possible, in the given instance, for each machine to obtain its maximin-shares simultaneously, resulting in a fair division of jobs in the mixed-manna setting. THis can be extended to solve the \addmms and \wfmms approximations. Further, by embedding the algorithms in a binary search framework, we can even search for the optimal values for \multmms\probName.

Meta-Theorem~\ref{meta:2} states that the minimal parameterizations 
of six parameters that are compatible with fixed parameter tractability, $\{n, m, \#d, p_{\max}, v_{\max}, g_{\max}\}$,  are the three described.
We prove this in Section~\ref{sec:hardness} by giving 
NP-hardness reductions for any stricter parameterizations.

\section{FPT algorithm with parameter \texorpdfstring{$n$}{n}}\label{sec:n}

Let us start with the smallest parameter set that gives an FPT algorithm for each of the problems, namely, parameterizing just by the number of jobs $n$. We first present a dynamic program that decides whether the targeted MMS values $\fairValue_i$ for each machine $m_i$ are achievable. A similar approach has been used in~\cite{JansenLL13}. This subroutine is then used to compute the MMS values for each machine, and can be adapted to solve the \addmms \probName{} and the \multmms \probName. For the \wfmms \probName, we cannot guess the MMS values based on the approximation guarantee because they are dependent on each other. Consequently, we need a more powerful, but slower approach: we construct a weighted graph that assigns the value of each possible bundle to machines, and solve the maximum matching problem on this graph efficiently.

\begin{theorem}
    When parameterized by $n$, we can decide the \probName{} for the targeted $MMS$-value $\fairValue_i$ for every machine $m_i$ in FPT time $2^{O(n)}\cdot \text{poly}(|I|)$. \label{theo:FPTdpn}
\end{theorem}
\begin{proof}
    We provide a dynamic program that calculates the MMS of all machines as follows: we construct a table $T$ containing $m$ rows and one column for every possible subset of the $n$ jobs. Table $T$ has size $m\times2^n.$ A cell $A[k,S]$ in $T$, representing the first $k$ machines and a subset of jobs $S$, receives value $A[k,S]=true$ if there is a schedule assigning a job set $\jobs_i$ to machine $m_i$ for $i \in [k]$, such that $\bigcup_{i=1}^k\jobs_i=S$ and $\jobs_i \cap \jobs_j = \emptyset$ for all $i \neq j$ and $v_i(\jobs_i)\geq \fairValue_i$ for all $i\in [k]$. Otherwise, we assign the value $A[k,S]=false$.
    
    We introduce a decision function $d_i: S \rightarrow \{$true, false$\}$ for every machine $m_i$ and subset $S\subseteq \jobs$ of jobs as $$d_i(S)=\begin{cases}
        \text{true} & \text{if} \sum_{j_t \in S} v_i(j_t) \geq \fairValue_i \text{\, and \,} S \text{\, is feasible}\\
        \text{false} & \text{otherwise.}
    \end{cases}$$

    We initialize the first row of $T$ with $A[1,S] = d_1(S)$ for all $S\subseteq \jobs$. We fill the table row-wise by calculating the value of $A[k,S]$ for $k\geq2$ as: $$
    A[k,S] = \begin{cases}
        \text{true} & \exists S' {\, s.t.\,} A[k-1,S\setminus S']=\text{\,true\,} \text{\,and\,}d_k(S') = \text{\,true}\\
        \text{false} & \text{otherwise.}
    \end{cases}$$
    
 For a single machine, the functions $d_i$ clearly output the correct result. Considering $k\geq 2$ and a subset of jobs $S$, the recurrence equation checks whether there is a partition $\{S', S\setminus S'\}$ such that $S'$ is feasible for machine $m_k$ and the remaining jobs $S\setminus S'$ can be feasibly assigned to the remaining $k-1$ machines (using entry $A[k-1,S\setminus S']$ in Table $T$). 

 Calculating the value for one cell takes time $2^n \cdot n$ as we need to try all possible subsets $S' \subset S$, calculate $d_k(S')$, and look up the value for $A[k-1,S\setminus S']$. Given that $T$ has a size of $2^n \times m$, we get an overall running time of $2^{O(n)}\cdot m$.
    
As the input has to encode the valuation functions of all machines, even if $m > n$, we have $|I|= \Omega(m).$ Therefore, the algorithm has a running time of $2^{O(n)}\cdot \text{poly}(|I|).$
\end{proof}

The dynamic program formulation given above can be implemented via a binary search framework to compute the MMS value of a machine. However, we improve the running time by adapting the formulation of the dynamic program to calculate a machines' MMS value directly. 
\begin{lemma}
    The MMS value $\MMSi$ for a machine $m_i$ can be computed in time $2^{O(n)} \cdot\text{poly}(|I|)$. \label{lem:dpnmms}
\end{lemma}
\begin{proof}
    We adapt the dynamic program from \cref{theo:FPTdpn} to calculate the MMS values instead of just checking for feasibility. Now, a cell $A[k,S]$ represents the optimal MMS value when scheduling only jobs in the subset $S$ on the first $k$ machines. As we are interested in the MMS value for machine $i$, we set the valuation function of all machines to $v_i$.
    
    First, we introduce a function $f: S \rightarrow \mathbb{N} \cup -\infty$ that calculates the MMS value for a machine with valuation function $v_i$ for a subset of jobs $S$. If the set is feasible, this is simply the sum of valuations of jobs in the subset $S$. However, if $S$ is not a feasible schedule due to the jobs deadlines, we assign a value of $-\infty$ to denote this, yielding$$f_i(S)=\begin{cases}
        \sum_{j_t \in S} v_i(j_t) &\text{if $S$ can be scheduled feasibly on a machine}\\
        -\infty &\text{else.}
    \end{cases}$$

    As before, we initialize the table by its first row as $A[1,S] = f_i(S)$ for all subsets $S$ of the jobs.
    
    We fill in Table $T$ row-wise for $k\geq 2$: $$A[i,S] =
           \max_{S'\subseteq S}\min \{A[i-1,S\setminus S'],f_i(S')\}.
    $$
 
    Overall, the MMS value of machine $m_i$ is the value in cell $A[m,\jobs]$. If this value is $-\infty$ then there exists no feasible schedule of the $n$ jobs on $m$ machines, which we exclude in our problem setting.

Similar to before, calculating the value for one cell takes time $2^n \cdot n$ as we must try all possible subsets $S' \subset S$, calculate $f(S')$, and look up the value for $A[k-1,S\setminus S']$. Given that $T$ has a size of $2^n \times m$, we get an overall running time of $2^{O(n)}\cdot m$, which can be bounded by $2^{O(n)}\cdot \text{poly}(|I|)$ as $|I|= \Omega(m)$.
\end{proof}

Applying Lemma~\ref{lem:dpnmms} for each $i \in [m]$, we can calculate the MMS value for each machine $m_i$. 
We use the values $\MMSi$ as the input values to the algorithm given in \cref{theo:FPTdpn}. This immediately determines if there exists a schedule that assigns each machine a bundle with a valuation of at least its MMS value. If this is impossible, we have to approximate these values. We can compute either an additive $\delta$ or multiplicative $\alpha$ approximation using a framework in which to embed the algorithm given above.

\begin{lemma}
We can solve
\begin{enumerate}
    \item the \addmms \probName{} in time $2^{O(n)}\cdot \text{poly}(|I|)$,
    \item the \multmms \probName{} in time $2^{O(n)}\cdot \text{poly}(|I|)$.
\end{enumerate}
    \label{lem:nAppdp}
\end{lemma}
\begin{proof}
    The optimal additive approximation is calculated via a binary search. We test for the smallest value $\delta$ such that there exists a schedule where each machine $m_i$ has at least valuation $\MMSi-\delta$ using \cref{theo:FPTdpn}. The bounds in which we need to search for $\delta$ are $[0,nv_{\max}].$ A binary search runs logarithmically in these values. Overall, this yields a running time of $2^{O(n)}\cdot \text{poly}(|I|)\cdot \log(nv_{\max})=2^{O(n)}\cdot \text{poly}(|I|).$

    For a multiplicative approximation, first observe that every machine can only receive an integer valuation due to all valuations being integer. As such, there is a discrete set of possible approximation factors $\alpha$ for every machine. Each of these $\alpha$ is given by an integer value lower than the MMS value of this machine, with the minimum being at $-nv_{\max}.$ Thus, there are $O(nv_{\max})$ possible factors $\alpha$ for each machine. Combining the possible approximation factors of all machines leads to a set of $O(mnv_{\max})$ many approximation factors. A binary search over this set then takes $O(\log(mnv_{\max}))$ many steps to arrive at the optimal approximation factor. We again set the values $\fairValue_i$ according to the binary search as $\alpha\cdot \MMSi$ and check for feasibility via \cref{theo:FPTdpn}. The running time is bounded by $2^{O(n)}\cdot \text{poly}(|I|)\cdot \log(mnv_{\max})= 2^{O(n)}\cdot \text{poly}(|I|)$.
\end{proof}
Finally, we tackle the most complex objective, that of the \wfmms \probName. The dynamic programming formulation given above cannot tackle this objective, as the guesses for the desired MMS values depend on each other, making a feasibility check within a binary search impossible. We can still solve this objective in a different manner, by constructing a graph representing the allocation of subsets of jobs and their associated valuations, and solving a minimum weight perfect matching problem. 
\begin{lemma}
    We can solve the \wfmms \probName{} in time $O(n^nm^3)$.
    \label{lem:wfn}
\end{lemma}
\begin{proof}
    We aim to minimize $\sum_{i}^{m} \text{MMS}_i-\fairValue_i= \sum_{i}^{m} \delta_i$, i.e., the sum of differences of the realized valuation $\fairValue_i$ for a machine $m_i$ and its MMS value $\MMSi$.
    
    We give an algorithm that considers every possible partition of the jobs into bundles, and their possible assignments to the machines. Exhaustively checking all these possibilities gives the desired running time, that is $O(n^n)$, if $m<n$. However, this approach does not guarantee FPT running time if $m > n$. To handle this case, we transform the problem to that of finding a matching in a bipartite graph that represents assigning bundles to machines.

    We start by constructing every feasible partition, i.e.\ every partition of the $n$ jobs into at most $\min\{n,m\}$ bundles such that each bundle satisfies the deadlines. Note that $n$ jobs cannot be partitioned into more than $n$ non-empty bundles, nor do we need to consider more than $m$ bundles if $m<n$. So there are at most $n^n$ of those partitions. For each of these partitions, we construct a bipartite graph. Let $\sigma$ be a partition with $k$ bundles. For $\sigma$, we construct a graph $G=(\{A,B\},E)$ as follows: the set of vertices $A$ contains $m$ vertices that represent the machines. For convenience, for a machine $m_i$, we also name the vertex $m_i$. The second set of vertices $B$ contains $k$ vertices that each represent a bundle in $\sigma$ and $m-k$ dummy vertices. Next, we construct the edges between $A$ and $B$.
    Edges incident to dummy vertices have weight zero. Edges incident to a bundle vertex set their weights according to the valuations of the bundles to the machines. In particular, the weight of the edge $\{m_i,\jobs_\ell\}$ between machine vertex $m_i$ and bundle vertex $\jobs_\ell$ is $\max\{\MMSi-v_i(\jobs_\ell), 0\}$. In this weighted graph, a solution to the \wfmms \probName{} corresponds to finding a matching of cardinality $m$ that has minimum weight. This is tractable with algorithms for the balanced assignment problem. We can solve this problem applying the Hungarian algorithm using Fibonacci heaps~\cite{FredmanT87}. The running time of this algorithm is $O(|E||V|+|V|^2\log(|V|))=O(m^2m+m^2\log m) = O(m^3).$ Doing this for every configuration yields a running time of $O(n^nm^3).$  
\end{proof}

\begin{theorem}
    For the \probName{} parameterized by $\{n\}$ we can in FPT time
    \begin{enumerate}
            \item[(i)] compute the maximin share value of every machine,
        \item[(ii)] determine whether or not a feasible allocation exists given targeted valuations for each machine,
        \item[(iii)] solve the \multmms, \addmms{} and \wfmms \probName.
    \end{enumerate}
    \label{theo:nobj}
\end{theorem}

We complement these results by showing a lower bound on the computational complexity of the \probName{} proving our above algorithms are essentially tight.

The exponential time hypothesis, formulated by Impagliazzo and Paturi~\cite{ImpagliazzoP99}, is a computational hardness assumption. It states that the NP-complete problem 3-SAT cannot be solved in time $2^{o(\ell)}\text{poly}(|I|)$, where $\ell$ is the number of variables in the 3-SAT formulation. Under this assumption, we can further determine similar running time lower bounds for problems that 3-SAT can be reduced to. A direct consequence of the ETH is that \partition{} cannot be solved in time $2^{o(k)}\text{poly}(|I|)$ for $k$ input numbers in the instance of \partition{}. We use this to show that the \probName{} cannot be solved in time of $2^{o(n)}\text{poly}(|I|)$, where $n$ is the number of jobs in the instance. 
\begin{theorem}
    Assuming ETH, the \probName{} cannot be solved in time $2^{o(n)}\text{poly}(|I|)$. \label{thm:ethLB}
\end{theorem}
\begin{proof}
    We show a reduction from \partition{} to an instance of the \probName. Let the \partition{} instance consist of $n$ numbers $s_1, \dots, s_n$. Let $t = \sum_{i\in [n]} s_i/2$ be the desired value for each partition.

    Let $m=2$. For every $s_i$, $i \in [n]$, we create a job $j_i$ with processing time $p_i=s_i$, deadline $d_i=t$ and valuation $v_1(j_i) = v_2(j_i) = 0$, i.e., we only care about a feasible schedule. 

    A solution $S^*=(S_1,S_2)$ of \partition{} directly implies the desired schedule for the \probName{}, as $\sum_{i \in S_1}s_i= t =\sum_{k \in S_2}s_k$, all jobs in both sets finish before their deadline. Hence, assigning the corresponding jobs of $S_1$ to $m_1$ and respectively, $S_2$ to $m_2$ yields a schedule where no job is late.

    For the other direction, a schedule where no job finishes late implies that there exists a partition of jobs into sets $J_1,J_2$ assigned to $m_1,m_2$, respectively, such that $\sum_{j\in J_1}p_j \leq t$ and $\sum_{k\in J_2}p_k \leq t$. This directly implies that $\sum_{j\in J_1}p_j= \sum_{j\in J_1}s_j= t = \sum_{k\in J_2}s_k = \sum_{k\in J_2}p_k$, so we have a feasible solution of \partition. 

    Assume now that there exists an algorithm $A$ that solves the \probName{} in time $2^{o(n)}\text{poly}(|I|).$ Then, we can transform any instance of \partition{} into an instance of the \probName, using the transformation above. This transformation is computable in polynomial time and creates one job per number $s_i$ in the input of \partition. Thus, solving the \probName{} using algorithm $A$ yields an algorithm that solves partition in time $2^{o(n)}\text{poly}(|I|).$ This contradicts the exponential time hypothesis.
\end{proof}

\section{FPT algorithm for group instances with parameters \texorpdfstring{$m, g_{\max},v_{\max}$}{m,gmax, vmax}}\label{sec:mvnmax}
In this section, we consider the case where the jobs can be partitioned into
$\lambda$ groups, which we call {\em layers}, with the largest group containing $g_{\max}$ jobs. (If there are no non-trivial groupings 
of the jobs then $\lambda=1$ and $g_{\max}=n$.)
We now parameterize by the number of machines $m$, the maximum number of jobs in a layer $g_{\max}$, and the largest absolute valuation $v_{\max}$. 
For motivation, think of these layers as separate days in a work schedule. While we could plan the assignment of jobs to machines for each single day separately, we desire that summed over all days, each machine receives a fair distribution of the jobs. As such, the core challenge for this setting derives from finding a combination of these independent daily schedules that are collectively fair for every machine. We show each variant is solvable in FPT time by constructing an $n$-fold integer program with appropriate dimensions.

First, we calculate all feasible schedules, i.e.\ assignments of jobs to machines, in any layer by enumeration. We denote by $\mathcal{S}_k$ the set of all schedules for layer $k$. Let $x_j^k \in \{0,1\}$ be an indicator variable equal to $1$ if and only if the schedule $j$ is chosen for layer $k$. Furthermore, let $v_i^{j,k}$ be the valuation machine $m_i$ has for schedule~$j$ in layer $k$.
Locally, the core idea in our $n$-fold IP is to check that there is a schedule for every layer. Globally, we ensure that the chosen set of bundles gives at least the desired targeted MMS value for every machine. Specifically,
we have the program
\begin{align}
    &\sum_k^\lambda\sum_j^{|\mathcal{S}_k|}x_j^k\cdot v_i^{j,k}\geq \fairValue_i &\forall i\in [m]\\
    &\sum_j^{|\mathcal{S}_k|}x_j^k=1 &\forall k \in [\lambda]
\end{align}
The constraints of Type~(2) ensure that there is exactly one schedule selected for every layer. As we only consider feasible schedules, this ensures that all jobs are feasibly packed. The constraints of Type~(1) ensure that every machine has at least the targeted MMS value $\fairValue_i$ over all bundles (defined by the selected schedule). The constraints of Type $(1)$ are the global constraints, while Type $(2)$ are local. 

\begin{lemma}\label{lem:nmaxm}
    If there is a feasible schedule $\sigma$ for the \probName{} with maximin share value $\fairValue_i$ for each machine $m_i$, then there is a solution for the above $n$-fold. Further, every solution to the $n$-fold corresponds to a feasible schedule over all layers with a maximin share value of at least $\fairValue_i$ for each machine~$m_i$.
\end{lemma}
\begin{proof}
    Let $\sigma$ be the feasible schedule for an instance $I^*$. Denote by $\sigma_k$ the schedule for layer~$k$. Now, for each layer $k$ proceed as follows: the schedule $\sigma_k$ must exist in the list of all feasible schedules $\mathcal{S}_k$. Denote it by $j$, and set the corresponding decision variable $x_j^k$ to one. Set $x_j^k = 0$ for $i \neq j$. This gives the solution to the above $n$-fold. 
    Clearly, all constraints of Type $(2)$ are satisfied, as we selected exactly one schedule $j$ per layer $k$. Furthermore, constraints of Type $(1)$ are also satisfied, as we copied the exact assignment of $\sigma$, which achieves a maximin share value of at least $\fairValue_i$ for every machine.

    To see the other direction, we use the one schedule $j$ with $x_j^k=1$ for each layer as the chosen schedule. This directly implies a feasible schedule with value at least $\fairValue_i$ for each machine $m_i$.    
\end{proof}

\paragraph*{Solving the $n$-fold}
We solve the $n$-fold using the algorithm by Cslovjecsek et al.~\cite{CslovjecsekEHRW21}. To calculate the running time, we need to bound the parameters $r,s,t$ and $\Delta$, representing the size of the blocks (matrices) and the largest entry in the $n$-fold respectively. The largest possible valuation, and therefore the largest entry $\Delta$ is $g_{\max}\cdot v_{\max}$. The number of global constraints $r$ is $m$. We have one local constraint for each block, so $s=1$. Finally, the number of constraints in each block is the amount of schedules for a layer. There are $g_{\max}^m$ possible partitions. Assigning all bundles of a partition to $m$ machines yields $m!\leq m^m$ possible assignments. As such, the total number of schedules for any layer can be bounded by $g_{\max}^m \cdot m^m$, which corresponds to number of variables, $t$, in a block. Putting this together yields the running time 

$$(\lambda g_{\max}^m m^m)^{(1+o(1))}2^{O(m)}(mg_{\max}v_{\max})^{O(m^2+1)},$$

which is FPT time for the given parameters $\{m, v_{\max}, g_{\max}\}$.

\begin{theorem}
    When parameterized by $\{m, g_{\max}, v_{\max}\}$, we can solve the \probName{} for a targeted MMS value $\fairValue_i$ for every machine in FPT time\\  $(\lambda g_{\max}^m m^m)^{(1+o(1))}2^{O(m)}(mg_{\max}v_{\max})^{O(m^2+1)}.$
    \label{theo:FPTLayers}
\end{theorem}
\begin{proof}
    We set up and solve the above $n$-fold and construct the final feasible schedule with the targeted MMS values as described in Lemma~\ref{lem:nmaxm}. Overall, the running time is $$(\lambda g_{\max}^m m^m)^{(1+o(1))}2^{O(m)}(mg_{\max}v_{\max})^{O(m^2+1)}.\qedhere$$    
\end{proof}

Then, using this information, we solve our three optimization problems. First, we must calculate the maximin share values $MMS_i$ of each machines $m_i$. We accomplish this using the same above $n$-fold.

\begin{lemma}
    We can compute $\MMSi$ for each machine in time $O(m \cdot \log(\lambda g_{\max}v_{\max})\cdot \rtNfold).$
    \label{lem:mnvMMS}
\end{lemma}
\begin{proof}
To calculate the maximin share values $\MMSi$ of a single machine $m_i$, we first replace the valuation of every machine with that of $m_i$. In this setting, the lowest value for $\fairValue_i$ corresponds to the maximin share value of this machine. We find this value via a binary search between the largest possible value $m_i$ can obtain, which is $\lambda g_{\max}v_{\max}$, and 0. 

Repeating this binary search for every machine requires $m \cdot \log(\lambda g_{\max}v_{\max})$ iterations of the $n$-fold to calculate the maximin share values of all machines. This yields an overall running time of $m \cdot \log(\lambda g_{\max}v_{\max})\cdot \rtNfold$.
\end{proof}
Given the maximin share values for each machine $m_i$, we can now adapt our $n$-fold to solve our optimization problems.
\begin{lemma}
We can solve the \multmms \probName{} in time $O(m \cdot\log(\lambda g_{\max}v_{\max}))\cdot \rtNfold$.

 \label{lem:mnvApp}
\end{lemma}
\begin{proof}
     If there does not exist a feasible schedule that assigns every machine $m_i$ a valuation of at least $\MMSi$, we need to find the best feasible multiplicative approximation. As explained in \cref{lem:nAppdp}, there is only a discrete set of possible approximation factors $\alpha$ for every machine as the valuations are integral. The minimum of these values is at least $-(\lambda g_{\max} v_{\max}).$ Therefore, we have $O(\lambda g_{\max} v_{\max})$ many possible approximation factors for each machine. Then, combining all these possible approximation factors into a single set yields at most $O(m\lambda g_{\max} v_{\max})$ many elements, over which we conduct our binary search. This takes $O(\log(m\lambda g_{\max} v_{\max}))$ many steps.
     
     Whenever an approximation factor $\alpha$ is selected, all $\fairValue_i$ are scaled by that factor. Any values that are non-integer as a result get rounded up, i.e.\ the first set of constraints transform to $\sum_k^\lambda\sum_j^{|\mathcal{S}_k|}x_j^k\cdot v_i^{j,k}\geq \lceil\alpha\MMSi\rceil$ for every machine with positive maximin share values. If a machine has negative maximin share values, we need to ensure that $\sum_k^\lambda\sum_j^{|\mathcal{S}_k|}x_j^k\cdot v_i^{j,k}\geq \lceil\frac{1}{\alpha}\MMSi\rceil.$ 
     
     In total, we have to calculate our $n$-fold once for every value of the binary search. The binary search computes in time $O(\log(m\cdot g_{\max}v_{\max})),$ so the total running time can be bounded by $O(\log(m\cdot g_{\max}v_{\max}))\cdot \rtNfold.$ This is dominated by the running time of \cref{lem:mnvMMS}, yielding the desired running time. 

\end{proof}

Computing the best additive approximation, in contrast to the previous objective, can be expressed as a linear objective function and hence be embedded into $n$-fold formulation. We elaborate why this approach does not work for \multmms \probName{} after the next lemma. 

\begin{lemma}
    We can solve the \addmms\probName{} problem in time \rtNfold.
    \label{lem:addMMSadapt}
\end{lemma}
\begin{proof}
    Approximating the minimal additive distance to the optimal MMS values can be expressed as a linear objective function of the $n$-fold. We add a variable $z$ that represents the maximum distance of all machines to their MMS value. The objective is then to minimize this $z.$ The remaining variables and values are defined as above. The modified $n$-fold is the following:
\begin{align}
& \min  z \\
    &\sum_k^\lambda\sum_j^{|\mathcal{S}_k|}x_j^k\cdot v_i^{j,k}\geq \MMSi - z &\forall i\in [m]\\
    &\sum_j^{|\mathcal{S}_k|}x_j^k=1 &\forall k \in [\lambda]
\end{align}

As the $z$ variable only occurs in the global variables, introducing it does not change the structure of the $n$-fold. The number of columns in each block is only increased by 1. Therefore, we can solve this $n$-fold in time $\rtNfold$.
\end{proof}
We remark that computing the optimal multiplicative approximation factor $\alpha$ in the manner above does not work as it cannot be expressed as a linear objective function to an $n$-fold over integral variables since $\alpha\in \mathbb{R}$.

Similarly, we tackle the most complex objective, that of the
optimal welfare maximin shares problem.

\begin{lemma}
    The \wfmms\probName{} is solvable in time $\rtNfold$.   
    \label{lem:mnvDist}
\end{lemma}
\begin{proof}
Minimizing $\delta_i=\text{MMS}_i-\fairValue_i$ for a machine $m_i$, can again be expressed as a linear function. Let variables $z_i \in [0, \dots, \MMSi]$, $i~\in [m]$, denote the distance to the desired $\MMSi$ value. We aim to minimize the sum of all $z_i$. The remaining variables and values are defined as above. Our altered $n$-fold is now as follows:

\begin{align}
& \min \sum_{i=1}^m z_i \\
    &\sum_k^\lambda\sum_j^{|\mathcal{S}_k|}x_j^k\cdot v_i^{j,k}\geq \MMSi - z_i &\forall i\in [m]\\
    &\sum_j^{|\mathcal{S}_k|}x_j^k=1 &\forall k \in [\lambda]
\end{align}

Introducing the $z_i$ variables does not change the structure of the $n$-fold. The number of columns in each block is only increased by $1$. Hence, we can also solve this $n$-fold in time $\rtNfold$.

\end{proof}

We may collectively summarize the results of this section in the following theorem.
\begin{theorem}
    For the \probName{} parameterized by $\{m,g_{\max},v_{\max}\}$ we can in FPT time
    \begin{enumerate}
            \item[(i)] compute the maximin share value of every machine,
        \item[(ii)] determine whether or not a feasible allocation exists given targeted valuations for each machine,
        \item[(iii)] solve the \multmms, \addmms\  and wf-\probName.
    \end{enumerate}
    \label{theo:mnvobj}
\end{theorem}

\section{FPT algorithm with parameters \texorpdfstring{$m, \#d, v_{\max}, p_{\max}$}{m,\#d,vmax, pmax}}\label{sec:mdvp}
We now present our third FPT algorithm for the parameterization $\{m, \#d, v_{\max}, p_{\max}\}$.
The critical part is constructing an appropriate $n$-fold integer program. Let $x_i^t\in \{0,1\}$ be an indicator variable describing whether or not job $j_t$ is assigned to machine $m_i$. 
We denote the $k$-th deadline (in order) by $D_k$ where $k\in [\#d]$.
As usual, $\fairValue_i$ is the targeted MMS value we need to exceed for machine $m_i$. Again, for now we assume these values are given and present an algorithm that either computes a schedule achieving these targets, or certifies that such a schedule does not exist. We use the same 
routine later to actually compute the true MMS values for each machine, and then to optimize our targets.
The relevant $n$-fold IP is defined as follows.
\begin{align}
    \sum_{ j_t \in \jobs : d_t \leq D_k}x_i^t p_t &\leq D_k &\forall k \in [\#d], i \in [m]\\
    \sum_{j_t \in  \jobs}x_i^t v_i^t&\geq \fairValue_i &\forall i\in [m]\\
    \sum_{i\in [m]} x_i^t &= 1 &\forall t \in [n]
\end{align}
The constraints of Type $(1)$ ensure that every job assigned to a machine finishes before its deadline. We check this for every unique deadline $D_k$ on every machine $m_i$. This attests a feasible schedule. The constraint of Type $(2)$ ensure that the targeted MMS values for each machine are achieved. Here $(1)$ and $(2)$ are our global constraints.
The Type $(3)$ constraints ensure that every job is assigned to exactly one machine. These are our local constraints and apply for every job $j_t$.
\begin{lemma}
    If there is a feasible schedule $\sigma$ for the \probName{} with maximin share value $\fairValue_i$ for each machine $m_i$, then there is a solution for the $n$-fold. Further, every solution to the $n$-fold corresponds to a feasible schedule over all layers with a maximin share value of at least $\fairValue_i$ for each machine~$m_i$.
    \label{lem:mdvpnfold}
\end{lemma}
\begin{proof}
    Given a schedule for an instance $I^*$, each job is assigned to exactly one machine. Set the $x_i^t$ variables according to this assignment. As no job is executed late in schedule~$\sigma$, this assignment is feasible. Furthermore, since this assignment produced a maximin share value of $\fairValue_i$ for every machine $m_i$, every machine has at least that total amount of value assigned to it. Thus, constraints of Type $(2)$ are satisfied.
    As a result, if the $n$-fold has no feasible solution, no schedule with a maximin share value of $\fairValue_i$ for every machine $m_i$ can exist for the given instance.
    
    Given a solution to the above $n$-fold integer program, we can construct a schedule easily: assign each job $j_t$ to a machine $m_i$ according to the corresponding $x_i^t$ indicator variable. This placement is feasible by construction of the $n$-fold, and yields a schedule of value at least $\fairValue_i$ for every machine~$m_i$. 
\end{proof}

\begin{theorem}
When parameterized by $m, \#d, v_{\max}, p_{\max}$, we can solve the \probName{} for a targeted MMS value $\phi_i$ for every machine $m_i$ in FPT time $\#d\cdot m^{3+o(1)} \cdot\max\{p_{\max}, v_{\max}\})^{O((\#d^2\cdot m^3)}n^{1+o(1)}$.
\label{theo:FPTmdvp}
\end{theorem}
\begin{proof}
    We solve the above $n$-fold using the algorithm by Cslovjecsek et al.~\cite{CslovjecsekEHRW21}. To calculate the running time, we must bound the parameters $r,s,t$ and $\Delta$, representing the size of the blocks (matrices) and the largest entry in the $n$-fold, respectively. There is a constraint of Type~$(1)$ for every combination of deadlines and machines, a total of $\#d\cdot m$. Constraints of Type $(2)$ are introduced for every machine. Hence, $r=\#d\cdot m+m$. We have one local constraint for every job. As we have one block for every job, we have $s=1$ and $N = n$. We have one variable for every machine in each block, so $t=m$. Finally, the factors of the constraint matrix are $v_i(j_t)$ and $p_t$. Therefore, $\Delta=\max\{v_i(j_t),p_t \,|\, i\in [m], j\in [n]\}=\max\{p_{\max}, v_{\max}\}.$ This gives a running time of
\begin{align*}
   &  2^{O((\#d\cdot m+m)\cdot m^2)}((\#d\cdot m+m)\cdot m\cdot\max\{p_{\max}, v_{\max}\})^{O((\#d\cdot m+m)^2\cdot m + m^2)}(n\cdot m)^{1+o(1)} \\
   & = \#d\cdot m^{3+o(1)} \cdot\max\{p_{\max}, v_{\max}\})^{O((\#d^2\cdot m^3)}n^{1+o(1)}
\end{align*}
which is FPT time for the parameters $\{m, \#d, v_{\max}, p_{\max}\}$.
\end{proof}
Again, denote the running time of the $n$-fold as $\rtNfold$. Using the $n$-fold framework to compute the $\MMSi$ values and to solve 
our optimization problem follows a similar approach to the proofs of \cref{lem:mnvMMS,lem:mnvApp,lem:mnvDist}. 
As such, in the proofs that follow we simply highlight any modifications
required and the corresponding changes in running time.

We begin with the calculation of the maximin share values. Our approach mirrors that of Lemma~\ref{lem:mnvMMS}, but with changed bounds for the binary search. The maximal valuation of a single machine is now $nv_{\max}$. Other than this difference, the proof is identical to that of \cref{lem:mnvMMS}, so we omit it.
\begin{lemma}
    We can compute $\MMSi$ for each machine $m_i$ in time $O(m \cdot \log(nv_{\max}))\cdot \rtNfold$.
    \label{lem:mdvpMMS}
\end{lemma}

Given the maximin share values for each machine $m_i$, we can now adapt our $n$-fold to solve our optimization problems.
    
\begin{lemma}
 We can solve the \multmms \probName{} in time $O(m \cdot \log(nv_{\max}))\cdot \rtNfold$.
 \label{lem:mdvpApp}
\end{lemma}
\begin{proof}
The proof resembles Lemma~\ref{lem:mnvApp}. 
We now have $mnv_{\max}$ possible values for the multiplicative approximation, yielding the stated running time via binary search. Note that the running times are dominated by the time to find the maximin share value for each machine.
\end{proof}

\begin{lemma}
 We can solve the \addmms \probName{} and \wfmms \probName{} in time $\rtNfold$.
\end{lemma}
\begin{proof}
        We adapt the $n$-fold as in Lemma~\ref{lem:addMMSadapt} and in Lemma~\ref{lem:mnvDist} and solve it. This immediately gives us the desired running time of $\rtNfold$.
\end{proof}

We may collectively summarize the results of this section in the following theorem.
\begin{theorem}
    For the \probName{} parameterized by $\{m,\#d,v_{\max},p_{\max}\}$ we can in FPT time
    \begin{enumerate}
            \item[(i)] compute the maximin share value of every machine,
        \item[(ii)] determine whether or not a feasible allocation exists given targeted valuations for each machine,
        \item[(iii)] solve the \multmms, \addmms\  and wf-\probName.
    \end{enumerate}
    \label{theo:mdvpobj}
\end{theorem}

We conclude this section with a remark about $d_{\max}$, the maximum deadline.
We did not include $d_{\max}$ as one of our six parameters because $d_{\max} \ge \max (\#d, p_{\max})$. Of course, this observation implies that our
FPT algorithms extend for the parameterization $\{m,d_{\max},v_{\max}\}$.
Indeed, in this case Theorem~\ref{theo:mdvpobj} applies.

\section{Hardness Results}\label{sec:hardness}
In this section we prove that our FPT time algorithms use the minimal subsets of parameters necessary to ensure fixed-parameter tractability.
Specifically we prove that, for any strict subsets of our three parameterizations, our problems become NP-hard, even if the parameters have {\em constant} values. This implies no FPT algorithms exist for stricter parameterizations, assuming $NP \neq P$. 

Futher, note that $\#d \leq d_{\max}$. Thus, $d_{\max}$ is the \emph{weaker} parameterization (i.e., the larger number), and thereby the hardness in every of our results also holds if we parameterize by $\#d$ instead of $d_{\max}$.

\subsection*{The \probName{} is NP-hard for constant $d_{\max}, v_{\max}, p_{\max}$}
We begin by proving that even if $m, d_{\max}, p_{\max}$ are constants, the \probName{} is NP-hard. We provide a reduction from a variant of 3-SAT called 3-SAT'. Like in 3-SAT, we are given a formula $\rho$ in conjuctive normal form with clauses of size at most $k=3$. Additionally, each variable $x^i$ occurs at most 3 times, and each literal (a variable or its negation) occurs at most twice. Deciding the feasibility of such an instance is NP-complete~\cite{PapadimitriouY91}. The reduction is designed such that it is NP-hard to decide whether there exists a schedule such that every machine receives an MMS value of at least 0. This reduction is inspired by one presented for the Graph Balancing problem~\cite{EbenlendrKS08}. 

\begin{theorem}\label{thm:hard1}
    The \probName{} is NP-hard if $d_{\max}, v_{\max}, p_{\max} \in O(1)$.
\end{theorem}
\begin{proof}
Given a 3-SAT' formula $\rho$ with variables $x^i$, we construct an instance of the \probName. We introduce one \emph{literal-machine} for each literal, i.e., each variable is represented by two machines $m_{x^i}, m_{\neg x^i}$ for the variable and its negation. Further, we construct machines $m_\beta$ for each clause $\beta$ in $\rho$. 

We create two types of jobs. We have \emph{assignment-jobs} $j_{x^i}$ for each variable $x^i.$ These have processing time $p_{x^i}=2$ and deadline $d_{x^i}=2.$ These jobs have value $0$ for their corresponding variable-machines $m_{x^i}, m_{\neg x^i}$ and a value $-1$ for any other machine. 
Next, we create \emph{literal-jobs}. Let $\beta=(x^i\lor \ldots \lor x^{j})$ be a clause in $\rho$ with $k$ literals. Note that $k\leq 3$. We construct jobs $j_{\beta,x^\ell}$ for each clause $\beta$ in $\rho$ and its literals $x^\ell$. They have processing time $p_{\beta,x^{i}}=1$ and deadline $d_{\beta,i}$. These jobs have value $0$ to their clause machine $m_\beta$ and their literal machines $m_{\beta,x^{i}}$, and value $-1$ for any other machine. Finally, if $k<3,$ we add a \emph{dummy job} $j_{\beta,\ell}$ with deadline $d_{\beta,\ell}=2$ and processing time $p_{\beta,\ell}=3-k.$ This job has value $0$ for only its clause machine $m_\beta$ and value $-1$ for every other machine. Clearly, this transformation can be computed in polynomial time. Further, we have $v_{\max}=1, d_{\max}\le d_{\max}=2$ and $p_{\max}=2$;  thus, these parameters have constant size.

First, we show that a yes instance of 3-SAT' implies a schedule where no machine has a negative valuation. Suppose we have a satisfying assignment $\sigma$ of our formula $\rho.$ For a variable $x^{i}$, we assign the assignment-job $j_{x^{i}}$ to $m_{\neg x^{i}}$ if $x^i$ is true in the solution to 3-SAT'. Otherwise, we assign it to $m_{\neg x^{i}}.$ Next, we assign the literal-jobs. We assign a literal job $j_{\beta,x^{i}}$ to its clause machine $m_\beta$ if its assignment evaluates as false in the solution, i.e., if variable $x^i$ is assigned false in the solution we assign literal-jobs that have value zero to $m_{x^{i}}$ to their respective clause machines. On the other hand, literal-jobs with value 0 for machine $m_{\neg x^{i}}$ are assigned to that very machine. Finally, dummy jobs $j_{\beta,\ell}$ are assigned to their clause machine $m_{\beta}$. As $\sigma$ is a satisfying assignment, there are at most two jobs assigned to each literal machine, thereby fitting inside the deadline. Furthermore, literal jobs are only assigned to the assignment given by $\sigma,$ so they do not occupy the same machine as the assignment jobs. Therefore, they also all fit the deadline. Finally, the dummy jobs also fit inside the deadline, as they occupy exactly as much processing time as is missing for a true-evaluation of any clause. All jobs are only assigned to machines they have value $0$ to, so no machine has a negative valuation.

Now, we show that a yes instance of the \probName{} where every machine has valuation at least 0 implies a satisfying assignment $\sigma$ to $\rho.$ We set variables in $\rho$ such that $x^i$ is true if the assignment job $j_{x^i}$ is placed upon machine $m_{\neg x^{i}}$ and false otherwise. Now consider any arbitrary clause $\beta.$ As there are at most two jobs placed upon the clause machine $m_\beta$ and the remaining literal-job is assigned to its corresponding machine (otherwise it would exceed the deadline of two and not be a feasible schedule), at least one literal inside this clause must evaluate as true. If any literal-job were assigned to a machine other than its clause or its literal, it would lead to a valuation of -1 to that machine. As no job has a positive valuation, this cannot be compensated. Thus, every clause must evaluate to true in order for there to exist a schedule with valuation 0 to all machines. This yields a satisfying assignment $\sigma$ for the 3-SAT' problem.
\end{proof}

\subsection*{The \probName{} is NP-hard for constant $m, d_{\max}, p_{\max}$}
Next we study the case of parameterizing by $m, d_{\max}, p_{\max}$. We prove hardness by reducing from \partition{}. In the \partition{} problem, we are given a multiset $S$ of $n$ positive integers $s_1, s_2, \ldots, s_n$. The question to decide is whether there exists a partition of $S$ into two multisets $S_1, S_2$ such that $\sum_{i\in S_1}s_i=\sum_{j\in S_2}s_j$ holds. This problem is well-known to be NP-hard.

\begin{theorem}\label{thm:hard2}
    The \probName{} is NP-hard if $m, d_{\max}, p_{\max}~\in~O(1)$.
\end{theorem}
\begin{proof}
    We transform an instance of \partition{} to an instance of the \probName. Let $m=2$.
    Let $S$ be the multiset of integers and $s_i$ be the value of the $i$-th element. For every $s_i$, we create a new job $j_i$ with processing time $p_i=0$, deadline $d_i=0$, and a valuation of $v_1(j_i)=v_2(j_i)=s_i$. The question is whether there exists a schedule such that the optimal MMS value of both machines is exactly $\sum_{i\in S}s_i / 2$.

    A solution $S^*=S_1,S_2$ of \partition{} directly implies an optimal split for the \probName. We partition the jobs into $S_1$ and $S_2$ and assign each machine one such partition. As the valuation on both machines is identical to $s_i$, both machines are assigned the same value, optimizing their maximin share values. Furthermore, as all processing times and deadlines are equal to zero, both sets can feasibly be scheduled on a machine.

    For the other direction, an MMS value optimal schedule $\sigma$ implies a solution $S^*$ for \partition. Let $\sigma_1$ be the set of jobs assigned to the first machine and $\sigma_2$ the set of jobs on machine two. As $\sigma$ is optimal, both machines have the same maximin share values. Construct $S_1$ and $S_2$ by taking the numbers from the transformed sets $\sigma_1$ and $\sigma_2$ respectively. As both machines have the same maximin share values, $\sum_{i\in S_1}s_i=\sum_{j\in S_2}S_j$ must hold. Thus, this is a feasible solution to \partition{} proving it is a yes instance.
\end{proof}

\subsection*{The \probName{} is NP-hard for constant $m, p_{\max}, v_{\max}$}
In this section, we prove that even if $m, p_{\max}, v_{\max}$ are constants, the \probName{} is NP-hard. Our reduction is again from \partition{}. The main difference is that we need to utilize deadlines to assign batches of specific overall value to one machine. In other words, we mimic different values of the jobs by batches of jobs with the same specifics w.r.t. processing time and valuations.

\begin{theorem}\label{thm:hard3}
    The \probName{} is NP-hard if $m, p_{\max}, v_{\max} \in O(1)$.
\end{theorem}
\begin{proof}
    We transform an instance of \partition{} to an instance of the \probName.
    Let $m=2$. For every number $s_i$ in the \partition{} instance, we create $2 s_i$ jobs $j_1^i, \ldots, j_{2 s_i}^i$. The first $s_i$ jobs have processing times $p_t^i=10$, valuation function $v_1(j_t^{i}) = v_2(j_t^{i}) = 1$
    and deadlines $d_{t}^i=\sum_{k=1}^{i-1}10s_k+10t$ for the $t = 1, \ldots, s_i$. The remaining $s_i$ jobs have processing time $p_{s_i+1}^i=8, p_{t}^i={10}, p_{2s_i}^i = 12$ and deadlines $d_{s_i+1}^i=\sum_{k=1}^{i-1}10s_k+8,  d_{j_t^i} = \sum_{k=1}^{i-1}10s_k+8+10(\ell-1), d_{j^i_{2s_i}}\sum_{k=1}^{i-1}10s_k+10s_i$ for $t = s_i+1, \dots, 2s_i-1$. These jobs have $v_1(j_t^{i}) = v_2(j_t^{i}) = 0$ for $t = s_i+1, \ldots, 2s_i$.
    If $s_i = 1$, then only two jobs are created, both with processing time $10$, deadline $d_{s_i+1}^i=\sum_{k=1}^{i-1}10s_k+10$, and evaluation functions and remaining deadlines as described before. 
    
    Note that due to the nature of their deadlines, the first $s_i$ jobs $j_1^i, \ldots, j_{s_i}^i$ need to be placed onto the same machine, and the remaining $j_{s_i+1}^i, \ldots, j_{2 s_i}^i$ onto the other machine respectively. The question is whether there exists a schedule such that both machines get a MMS value of $\sum_{i=1}^n s_i/2$.
   
    We first show that a solution $(S_1, S_2)$ to \partition{} implies a schedule with MMS value of $\sum_{i=1}^n s_i/2$. We construct a schedule $\sigma$. For each element $s_i \in S_1$, we assign the set of jobs $j_1^i, \ldots, j_{s_i}^i$ to machine $m_1$. This yields an additive term of $s_i$ to the valuation of $m_1$. For each element $s_i \notin S_1$, we assign the set of jobs $j_{s_i+1}^i, \ldots, j_{2 s_i}^i$ to machine $m_1$, not changing the valuation. The remaining jobs go to $m_2$. As $(S_1, S_2)$ is a partition, the resulting schedule results in a valuation for both machines of $\sum_{i=1}^n s_i/2.$ Both of these machines can be scheduled without exceeding a single deadline. 

    Now we show that a yes instance of the \probName{} implies a yes instance of \partition. Let $\sigma$ be the schedule in which both machines have a valuation of $\sum_{i=1}^n s_i/2$. Due to the structure of the deadlines, this schedule must be separated into $n$ layers, each with size $10s_i$. Let us inspect the first block closer. W.l.o.g.\ assume that machine $m_1$ holds the first job of set $j_1^1$. This job has a processing time of 10 and a deadline of exactly 10, so it must be placed first inside the schedule. Then, the next jobs on $m_1$ in this schedule must also be $j_2^1, \ldots, j_{s_i}^1$. Hence, the jobs $j_{s_i+1}^1, \ldots, j_{2s_i}^1$ must be placed on the respective other machine. This pattern repeats for every layer, i.e., every number corresponding additive valuation term $s_i$ either being placed on $m_1$ or $m_2$ as a batch. Thus, the placement of batches with positive evaluation on $m_1$ yields the partition $S_1$, and the remaining numbers are in $S_2$ solving the \partition{} problem.
\end{proof}

\subsection*{The \probName{} is NP-hard for constant $m, \#d, v_{\max}$}
Next we study the case of parameterizing by $m, \#d, v_{\max}$. We prove hardness by reducing from \ecpartition{}. In the \ecpartition{} problem, we are given a multiset $S$ of $n$ positive integers $s_1, s_2, \ldots, s_n$. The question to decide is whether there exists a partition of $S$ into two multisets $S_1, S_2$ with $|S_1|=|S_2|=\frac12n$ such that $\sum_{i\in S_1}s_i=\sum_{j\in S_2}s_j$ holds. This problem is NP-hard.

\begin{theorem}\label{thm:hard4}
    The \probName{} is NP-hard if $m, \#d, v_{\max}~\in~O(1)$.
\end{theorem}
\begin{proof}
    We transform an instance of \ecpartition{} to an instance of the \probName. Let $m=2$.
    Let $S$ be the multiset of integers and $s_i$ be the value of the $i$-th element. For every $s_i$, we create a new job $j_i$ with processing time $p_i=s_i$, deadline $d_i=\frac12\cdot \sum_{j\in S} s_j$, and a valuation of $v_1(j_i)=v_2(j_i)=1$. Thus $\#d=1$ and $v_{\max}=1$. It is straight-forward to verify that there exists a feasible schedule such that the optimal MMS value of both machines is exactly $\frac12 n$ if and only if an equal-cardinality partition exists. 
\end{proof}
Theorems~\ref{thm:hard1}-\ref{thm:hard4} prove that the parameters $\{m, \#d, v_{\max}, p_{\max}\}$ form a minimal set to ensure fixed-parameter tractability. Similar hardness reductions apply for the parameters $\{m, g_{\max}, v_{\max}\}$. (The parameterizations of $\{g_{\max},v_{\max}\}$ and $\{m,v_{\max}\}$ are covered by \cref{thm:hard1,thm:hard4}. To show hardness of $\{m,g_{\max}\}$, we can utilize a similar reduction to \cref{thm:hard3}, where both jobs belonging to one input number $s_i$, one with valuation $s_i$ and the other with valuation $0$, of the \partition{} problem form a separate layer.) The parameter set $\{n\}$ is trivially minimal.

\section{Extensions to admit rejection of jobs}\label{sec:extensions}
In scheduling theory, jobs admitting deadlines often also permit the planner to schedule them after their respective deadline, see e.g.~\cite{LenstraShmoysEle20}. Doing so incurs some sort of penalty $w$, either a unit penalty per late job or some weight $w_t$ assigned to each job $t$. A natural question is whether our algorithms can be extended to admit rejections while still providing optimal results. This is indeed the case. In fact, only a slight modification to our algorithms is required to also admit late jobs. 

We introduce an additional machine $m_{\text{late}}$, on which we place the late jobs. The valuation function of this additional machine is always negative and reflects the penalties of the jobs, i.e.\ a job with a penalty $w$ has a valuation of $-w$ on $m_{\text{late}}$. Machine $m_{\text{late}}$ is the only machine on which we do not conduct our feasibility checks, i.e.\ we allow bundles that do not conform to their deadlines. Then, optimizing the maximin share values of machines can be computed as given above, and we can set a limit of penalties for late jobs by capping the maximum absolute value assigned to $m_{\text{late}}$. This adaptation allows us to solve problems on an entirely new axis: How large of a penalty must we incur for us to achieve an $\delta$-additive or $\alpha$-multiplicative maximin share value approximation? To compute this penalty value, we simply insert the targeted MMS values for all machines in the frameworks given in~\cref{theo:FPTdpn,theo:FPTLayers,theo:FPTmdvp}. Thus, we require a new parameter to bound the largest entry in our $n$-fold algorithms. This parameter is the maximum penalty value of a single job, which we denote by $w_{\max}$. Next, we conduct a binary search of the penalty value in bounds $[0,nw_{\max}].$ Thus, we require $O(\log(nw_{\max}))$ many iterations of any of the given algorithms to decide this question. Overall, this allows us to formulate the following theorem:

\begin{theorem}
 The $\times$-, \addmms, \wfmms MMS scheduling with deadlines and rejections problem is fixed parameter tractable when parameterized by (i) $n$, (ii) $\{m,v_{\max}, n_{max},w_{\max}\}$, (iii) $\{m,\#d,p_{\max},v_{\max},w_{\max}\}.$
\end{theorem}
\begin{proof}
    We restate the modifications to \cref{theo:nobj,theo:mnvobj,theo:mdvpobj} required for the results. We construct the additional machine $m_{\text{late}}$ and set its valuation function to be $-w_t$ for all jobs $j_t$ with $t \in [n].$ Then, we construct our dynamic program, graph or $n$-fold respectively. For each construction, we do not constrain feasibility on $m_{\text{late}}.$ For our graph algorithm, we thus construct \emph{every} bundle and connect it to $m_{\text{late}}.$ Thus, we now need to solve the unbalanced assignment problem, which is also tractable via the Hungarian algorithm. However, its running time depends on the number of edges and therefore increases to $O(n^n(m+1) (m+1)^2\log (m+1)),$ which is still FPT for parameter $n$. For our $n$-fold algorithms, we simply increase the number of feasible configurations on $m_{\text{late}}$. The core difference is the appearance of $w_{\max}$ as a possible largest parameter in the $n$-fold. Therefore, we need to include it in the list of parameterizations. Doing this allows us to reuse the frameworks given in Sections~\ref{sec:n}-\ref{sec:mdvp}.
\end{proof}

\section{Alternative Fairness Measures: Envy-Freeness}\label{sec:EF}
The reader may wonder whether FPT algorithms can be obtained for 
job scheduling with other classical fairness measures. We show now why 
this is not the case for envy-freeness. 
Let $(\jobs_1,\ldots,\jobs_m)$ be a feasible assignment of the set of jobs $\jobs$ to the machines. This assignment is {\em envy-free} if, for every machine $m_i$, $v_i(\jobs_i)\geq v_i(\jobs_\ell)$, for all $\ell \in [m]$. This concept is due to Foley~\cite{Fol67}. However, envy-free allocations need not exist
so several famous relaxations have been studied.
An allocation is EF-$1$ if, for any two machines $m_j, m_\ell$, there exists 
a job which we can remove from $\ell$'s allocation so that $m_j$ does not envy $m_\ell$. For the standard fair division problem, EF-$1$ allocations are guaranteed to exist~\cite{Bud11}.
 Similarly, an EF-$k$ allocation allows us remove up to $k$ jobs so that no machines envy each other, whereas an EF-$x$ allocation requires we remove the least valuable job.

However, even for the straightforward case of two machines $m_1$ and $m_2$ with identical evaluation functions, that is, $v_1(j_t) = v_2(j_t)$ for all jobs $j_t \in \jobs$, it can be shown via a simple example that no EF-$k$ allocation exists for any $k \leq n-2$, also ruling out the possibility of EF-$x$ allocations.

\begin{theorem}
    There exists no EF-$k$ (thus, also no EF-$x$) fair allocation with scheduling constraints, even for the case of just two machines, two different processing times, one deadline, and the same evaluation function, that is, $v_1(j_t) = v_2(j_t)$ for all $j_t \in \jobs$ which only has two distinct values.
\end{theorem}
\begin{proof}
Let $m=2$. We create $n$ jobs.
The first $n-1$ jobs have value 1 on both machines, i.e., $v_1(j_p) = v_2(j_p) = 1$, processing time $p_p = 1$, and deadline $d_p = n-1$ for $p = 1, 2, \dots, n-1$.
The last job has a value of $0$, i.e. $v_1(j_n) = v_2(j_n) = 0$, processing time $j_n = n-1$, and deadline $d_n = n-1$.

As such, the only feasible schedule must assign the first $n-1$ jobs onto one machine, and the remaining job onto the other machine. Let us assume w.l.o.g. that machine $m_2$ gets job $j_n$. Thus, machine $m_2$ is envious of machine $m_1$ unless all $n-1$ jobs on machine $m_1$ get removed.
\end{proof}

This daunting result has a simple explanation: the scheduling constraints inherently limit the number and structure of the feasible allocations (schedules). However, traditional fairness measures, such as EF-$k$ and EF-$x$, do not consider these constraints when assessing the fairness of allocations. 
\bibliography{ref}

\begin{thebibliography}{10}

\bibitem{bible}
In {\em The Bible}, chapter~13. Book of Genesis.

\bibitem{AG24}
H.~Akrami and J.~Garg.
\newblock Breaking the $\frac34$ barrier for approximate maximin share.
\newblock In {\em Proceedings of 35th Symposium on Discrete Algorithms (SODA)},
  2024.

\bibitem{AT20}
M.~Aleksandrov and T.~Walsh.
\newblock Online fair division: A survey.
\newblock In {\em Proceedings of 34th AAAI Conference on Artificial
  Intelligence (AAAI)}, pages 13557--13562, 2020.

\bibitem{AAB23}
G.~Amanatidis, H.~Aziz, G.~Birmpas, A.~Filos-Ratsikas, B.~Li, H.~Moulin,
  A.~Voudouris, and X.~Wu.
\newblock Fair division of indivisible goods: Recent progress and open
  questions.
\newblock {\em Artificial Intelligence}, 322:103965, 2023.

\bibitem{AMN17}
G.~Amanatidis, E.~Markakis, A.~Nikzad, and A.~Saberi.
\newblock Approximation algorithms for computing maximin share allocations.
\newblock {\em ACM Transactions on Algorithms}, 13(4):1--28, 2017.

\bibitem{AschenbrennerH07}
M.~Aschenbrenner and R.~Hemmecke.
\newblock Finiteness theorems in stochastic integer programming.
\newblock {\em Found. Comput. Math.}, 7(2):183--227, 2007.

\bibitem{Aziz20}
H.~Aziz.
\newblock Developments in multi-agent fair allocation.
\newblock In {\em Proceedings of 34th AAAI Conference on Artificial
  Intelligence (AAAI)}, pages 13563--13568, 2020.

\bibitem{ALM22}
H.~Aziz, B.~Li, H.~Moulin, and X.~Wu.
\newblock Algorithmic fair allocation of indivisible items: A survey and new
  questions.
\newblock {\em ACM SIGecom Exchanges}, 20(1):24--40, 2022.

\bibitem{ARS17}
H.~Aziz, G.~Rauchecker, G.~Schryen, and T.~Walsh.
\newblock Algorithms for max-min share fair allocation of indivisible chores.
\newblock In {\em Proceedings of 31st AAAI Conference on Artificial
  Intelligence (AAAI)}, pages 335--341, 2017.

\bibitem{BK20}
S.~Barman and S.~Krishnamurthy.
\newblock Approximation algorithms for maximin fair division.
\newblock {\em ACM Transactions on Economics and Computation (TEAC)},
  8(1):1--28, 2020.

\bibitem{BBN16}
B.~Bleim, R.~Bredereck, and R.~Niedereier.
\newblock Complexity of efficient and envy-free resource allocation fair
  allocation: few agents resources, or utility levels.
\newblock In {\em Proceedings of 25th International Joint Conference on
  Artifiical Intelligence (IJCAI)}, pages 102--108, 2016.

\bibitem{BH20}
E.~Bokanyi and A.~Hannak.
\newblock Understanding inequalities in ride-hailing services through
  simulations.
\newblock {\em Scientific Reports}, 10, 2020.

\bibitem{BL16}
S.~Bouveret and M.~Lematre.
\newblock Characterizing conflicts in fair division of indivisible goods using
  a scale of criteria.
\newblock {\em Autonomous Agents and Multi-Agent Systems}, 30(2):259–290,
  2016.

\bibitem{BK96}
S.~Brahms and D.~King.
\newblock {\em Fair Division: From Cake Cutting to Dispute Resolution}.
\newblock Cambridge University Press, 1996.

\bibitem{BKK19}
R.~Bredereck, A.~Kaczmarczyk, D.~Knop, and R.~Niedereier.
\newblock High-multiplicity fair allocation: Lenstra empowered by $n$-fold
  integer programming.
\newblock In {\em Proceedings of 19th ACM Conference on Economics and
  Computation (EC)}, pages 505--523, 2019.

\bibitem{Bud11}
E.~Budish.
\newblock The combinatorial assignment problem: Approximate competitive
  equilibrium from equal incomes.
\newblock {\em Journal of Political Economy}, 119(6):1061--1103, 2011.

\bibitem{ChenMYZ17}
L.~Chen, D.~Marx, D.~Ye, and G.~Zhang.
\newblock Parameterized and approximation results for scheduling with a low
  rank processing time matrix.
\newblock In {\em 34th Symposium on Theoretical Aspects of Computer Science,
  {STACS} 2017}, volume~66 of {\em LIPIcs}, pages 22:1--22:14. Schloss Dagstuhl
  --- Leibniz-Zentrum f{\"{u}}r Informatik, 2017.

\bibitem{CslovjecsekEHRW21}
J.~Cslovjecsek, F.~Eisenbrand, C.~Hunkenschr{\"{o}}der, L.~Rohwedder, and
  R.~Weismantel.
\newblock Block-structured integer and linear programming in strongly
  polynomial and near linear time.
\newblock In {\em Proceedings of 31st Symposium on Discrete Algorithms (SODA)},
  pages 1666--1681. SIAM, 2021.

\bibitem{CslovjecsekEPVW21}
J.~Cslovjecsek, F.~Eisenbrand, M.~Pilipczuk, M.~Venzin, and R.~Weismantel.
\newblock Efficient sequential and parallel algorithms for multistage
  stochastic integer programming using proximity.
\newblock In {\em 29th Annual European Symposium on Algorithms, {ESA} 2021},
  volume 204 of {\em LIPIcs}, pages 33:1--33:14. Schloss Dagstuhl ---
  Leibniz-Zentrum f{\"{u}}r Informatik, 2021.

\bibitem{DeLoera2008}
J.~{De Loera}, R.~Hemmecke, S.~Onn, and R.~Weismantel.
\newblock ${N}$-fold integer programming.
\newblock {\em Discrete Optimization}, 5(2):231--241, 2008.
\newblock In Memory of George B. Dantzig.
\newblock URL:
  \url{https://www.sciencedirect.com/science/article/pii/S1572528607000230},
  \href {https://doi.org/10.1016/j.disopt.2006.06.006}
  {\path{doi:10.1016/j.disopt.2006.06.006}}.

\bibitem{EPS22}
S.~Ebadian, D.~Peters, and N.~Shah.
\newblock How to allocate fairly easy and difficult chores.
\newblock In {\em Proceedings of 21st Conference on Autonomous Agents and
  Multiagent Systems (AAMAS)}, pages 372--80, 2022.

\bibitem{EbenlendrKS08}
T.~Ebenlendr, M.~Krc{\'{a}}l, and J.~Sgall.
\newblock Graph balancing: a special case of scheduling unrelated parallel
  machines.
\newblock In {\em Proceedings of 19th Symposium on Discrete Algorithms (SODA)},
  pages 483--490. {SIAM}, 2008.

\bibitem{EisenbrandHK18}
F.~Eisenbrand, C.~Hunkenschr{\"{o}}der, and K-M. Klein.
\newblock Faster algorithms for integer programs with block structure.
\newblock In {\em Proceedings of 45th International Colloquium on Automata,
  Languages, and Programming, (ICALP)}, volume 107 of {\em LIPIcs}, pages
  49:1--49:13. Schloss Dagstuhl --- Leibniz-Zentrum f{\"{u}}r Informatik, 2018.

\bibitem{FST21}
U.~Feige, A.~Sapir, and L.~Tauber.
\newblock A tight negative example for mms fair allocations.
\newblock In {\em Proceedings of 17th Conference on Web and Internet Economics
  (WINE)}, pages 355--372, 2021.

\bibitem{Fol67}
D.~Foley.
\newblock Resource allocation and the public sector.
\newblock {\em Yale Economic Studies}, 7(1):45--98, 1967.

\bibitem{FredmanT87}
Michael~L. Fredman and Robert~Endre Tarjan.
\newblock Fibonacci heaps and their uses in improved network optimization
  algorithms.
\newblock {\em J. {ACM}}, 34(3):596--615, 1987.

\bibitem{GT21}
J.~Garg and S.~Taki.
\newblock An improved approximation algorithm for maximin shares.
\newblock {\em Artificial Intelligence}, 300, 2021.
\newblock \#103547.

\bibitem{GavenciakKK22}
T.~Gaven\v{c}iak, M.~Kouteck{\'{y}}, and D.~Knop.
\newblock Integer programming in parameterized complexity: Five miniatures.
\newblock {\em Discret. Optim.}, 44(Part):100596, 2022.

\bibitem{GS20}
G.~Greco and F.~Scarcello.
\newblock The complexity of computing maximin share allocations on graphs.
\newblock In {\em Proceedings of 34th AAAI Conference on Artificial
  Intelligence (AAAI)}, pages 2006--2013, 2020.

\bibitem{GNC22}
A~Gupta, R.~Yadav, A.~Nair, A.~Chakraborty, S.~Ranu, and A.~Bagchi.
\newblock Fairfoody: bringing in fairness in food delivery.
\newblock In {\em Proceedings of 36th AAAI Conference on Artificial
  Intelligence (AAAI)}, pages 11900--11907, 2022.

\bibitem{HNN18}
T.~Heinen, N-T. Nguyen, T-T. Nguyen, and J.~Rothe.
\newblock Approximation and complexity of the optimization and existence
  problems for maximin share, proportional share, and minimax share allocation
  of indivisible goods.
\newblock {\em Autonomous Agents and Multi-Agent Systems}, 32(6):741--778,
  2018.

\bibitem{HemmeckeOR13}
R.~Hemmecke, S.~Onn, and L.~Romanchuk.
\newblock $n$-fold integer programming in cubic time.
\newblock {\em Math. Program.}, 137(1-2):325--341, 2013.

\bibitem{Hill87}
T.~Hill.
\newblock Partitioning general probability measures.
\newblock {\em The Annals of Probability}, 15(2):804--813, 1987.

\bibitem{ImpagliazzoP99}
Russell Impagliazzo and Ramamohan Paturi.
\newblock Complexity of k-sat.
\newblock In {\em {CCC}}, pages 237--240. {IEEE} Computer Society, 1999.

\bibitem{JansenKMR22}
K.~Jansen, K-M. Klein, M.~Maack, and M.~Rau.
\newblock Empowering the configuration-{IP}: new {PTAS} results for scheduling
  with setup times.
\newblock {\em Math. Program.}, 195(1):367--401, 2022.

\bibitem{JansenLL13}
K.~Jansen, F.~Land, and K.~Land.
\newblock Bounding the running time of algorithms for scheduling and packing
  problems.
\newblock {\em Bericht Des Instituts Für Informatik, Vol. 1302}, 2013.

\bibitem{JansenLR20}
K.~Jansen, A.~Lassota, and L.~Rohwedder.
\newblock Near-linear time algorithm for $n$-fold {ILP}s via color coding.
\newblock {\em {SIAM} J. Discret. Math.}, 34(4):2282--2299, 2020.

\bibitem{KnopK18}
D.~Knop and M.~Kouteck{\'{y}}.
\newblock Scheduling meets $n$-fold integer programming.
\newblock {\em J. Sched.}, 21(5):493--503, 2018.

\bibitem{KnopKLMO21}
D.~Knop, M.~Kouteck{\'{y}}, A.~Levin, M.~Mnich, and S.~Onn.
\newblock Parameterized complexity of configuration integer programs.
\newblock {\em Oper. Res. Lett.}, 49(6):908--913, 2021.

\bibitem{KnopKM20}
D.~Knop, M.~Kouteck{\'{y}}, and M.~Mnich.
\newblock Combinatorial $n$-fold integer programming and applications.
\newblock {\em Math. Program.}, 184(1):1--34, 2020.

\bibitem{KnopKM20b}
D.~Knop, M.~Kouteck{\'{y}}, and M.~Mnich.
\newblock Voting and bribing in single-exponential time.
\newblock {\em {ACM} Trans. Economics and Comput.}, 8(3):12:1--12:28, 2020.

\bibitem{KouteckyLO18}
M.~Kouteck{\'{y}}, A.~Levin, and S.~Onn.
\newblock A parameterized strongly polynomial algorithm for block structured
  integer programs.
\newblock In {\em Proceedings of 45th International Colloquium on Automata,
  Languages, and Programming, (ICALP)}, volume 107 of {\em LIPIcs}, pages
  85:1--85:14. Schloss Dagstuhl --- Leibniz-Zentrum f{\"{u}}r Informatik, 2018.

\bibitem{KMT21}
R.~Kulkarni, R.~Mehta, and S.~Taki.
\newblock Indivisible mixed manna: On the computability of {MMS + PO}
  allocations.
\newblock In {\em Proceedings of 22nd ACM Conference on Economics and
  Computation (EC)}, pages 683--684, 2021.

\bibitem{KPW18}
D.~Kurokawa, A.~Procaccia, and J.~Wang.
\newblock Fair enough: Guaranteeing approximate maximin shares.
\newblock {\em Journal of the ACM}, 65(2):1--27, 2018.

\bibitem{LenstraShmoysEle20}
J.~Lenstra and D.~Shmoys.
\newblock Elements of scheduling.
\newblock {\em CoRR}, abs/2001.06005, 2020.

\bibitem{LZB19}
N.~Lesmana, X.~Zhang, and X.~Bei.
\newblock Balancing efficiency and fairness in on-demand ridesourcing.
\newblock In {\em Proceedings of the 33rd Conference on Neural Information
  Processing Systems (NeurIPS)}, 2019.

\bibitem{LT20}
Z.~Lonc and M.~Truszczynski.
\newblock Maximin share allocations on cycles.
\newblock {\em Journal of Artificial Intelligence Research}, 69:613--655, 2020.

\bibitem{PapadimitriouY91}
C.~Papadimitriou and M.~Yannakakis.
\newblock Shortest paths without a map.
\newblock {\em Theor. Comput. Sci.}, 84(1):127--150, 1991.

\bibitem{PW14}
A.~Procaccia and J.~Wang.
\newblock Fair enough: Guaranteeing approximate maximin shares.
\newblock In {\em Proceedings of 15th ACM Conference on Economics and
  Computation (EC)}, pages 675--692, 2020.

\bibitem{RW98}
J.~Robertson and W.~Webb.
\newblock {\em Cake Cutting Algorithms: Be Fair if you can}.
\newblock A K Peters, 1998.

\bibitem{SDC23}
D.~Singh, S.~Das, and A.~Chakraborty.
\newblock Fairassign: Stochastically fair driver assignment in gig delivery
  platforms.
\newblock In {\em Proceedings of the 2023 ACM Conference on Fairness,
  Accountability, and Transparency (FAccT)}, page 753–763, 2023.

\bibitem{Ste48}
H.~Steinhaus.
\newblock The problem of fair division.
\newblock {\em Econometrica}, 16:101--104, 1948.

\bibitem{SBZ19}
T.~S\"{u}hr, A.~Biega, M.~Zehlike, K.~Gummadi, and A.~Chakraborty.
\newblock Two-sided fairness for repeated matchings in two-sided markets: a
  case study of a ride-hailing platform.
\newblock In {\em Proceedings of the 25th Conference on Knowledge Discovery and
  Data Mining (KDD)}, pages 3082--3092, 2019.

\bibitem{SJY22}
J.~Sun, H.~Jin, Z.~Yang, L.~Su, and X.~Wang.
\newblock Optimizing long-term efficiency and fairness in ride-hailing via
  joint order dispatching and driver repositioning.
\newblock In {\em Proceedings of the 28th Conference on Knowledge Discovery and
  Data Mining (KDD)}, pages 3950--3960, 2022.

\bibitem{Woe97}
G.~Woeginger.
\newblock A polynomial-time approximation scheme for maximizing the minimum
  machine completion time.
\newblock {\em Operations Research Letters}, 20(4):149--154, 1997.

\bibitem{ZCA18}
S.~Zoepf, S.~Chen, P.~Adu, and G.~Pozo.
\newblock The economics or ride-hailing: driver revenue, expenses and taxes.
\newblock Working paper, Center for Energy and Environmental Policy Research,
  2018.

\end{thebibliography}

\end{document}